\documentclass[11pt]{article}
\usepackage{times}
\usepackage{helvet}
\usepackage{courier}
\usepackage{graphicx}
\usepackage{subfigure}
\usepackage{fullpage}
\usepackage{tikz}
\usetikzlibrary{shapes.gates.logic.US,trees,positioning,arrows}
\usetikzlibrary{shapes.geometric}
\usetikzlibrary{positioning}

\tikzset{
  treenode/.style = {align=center, inner sep=1pt},
  ma/.style = {draw,treenode, shape border rotate=90, isosceles triangle,isosceles triangle apex angle=60, black, minimum width=8mm},
  mi/.style = {ma, shape border rotate=-90, font=2},
  ch/.style = {draw, treenode, circle, minimum width=8mm, black},
  lbl/.style = {fill=white,inner sep=1pt},
  no/.style = {draw, treenode, color=white, text=black}
}

\usepackage{amsmath, amssymb}
\usepackage{mathtools}
\usepackage[shortlabels]{enumitem}
\usepackage{comment}
\DeclareMathOperator*{\argmax}{arg\,max}

\DeclarePairedDelimiter\floor{\lfloor}{\rfloor}

\newcommand{\calS}{{\cal S}}
\newcommand{\calP}{{\cal P}}
\newcommand{\calN}{{\cal N}}
\newcommand{\calZ}{{\cal Z}}
\newcommand{\calA}{{\cal A}}
\newcommand{\calI}{{\cal I}}
\newcommand{\calT}{{\cal T}}
\newcommand{\calC}{{\cal C}}
\newcommand{\calBR}{{\cal BR}}

\newtheorem{definition}{Definition}

\newtheorem{claim}{Claim}

\newtheorem{example}{Example}
\newtheorem{theorem}{Theorem}
\newtheorem{lemma}{Lemma}
\newenvironment{proof}{\paragraph{Proof:}}{\hfill$\square$}

\newcommand{\eps}{\varepsilon}

\usepackage{pifont}
\newcommand{\cmark}{\ding{51}}%
\newcommand{\xmark}{\ding{55}}%

\title{Computation of Stackelberg Equilibria of Finite Sequential Games\footnote{The authors aknowledge support from
the Danish National
Research Foundation and The National Science Foundation of China
(under the grant
61361136003) for the Sino-Danish Center for the Theory of Interactive Computation and
from the Center for Research in Foundations of Electronic Markets (CFEM), supported by
the Danish Strategic Research Council.
Branislav Bosansky was also supported by the Czech Science Foundation (grant no. 15-23235S).
Simina Br\^{a}nzei
was also supported by ISF grant 1435/14 administered by the Israeli Academy of Sciences and Israel-USA Bi-national Science Foundation (BSF) grant  2014389 and the I-CORE Program of the Planning and Budgeting Committee and The Israel Science Foundation.
}}

\author{Branislav Bo{\v s}ansk{\' y}\footnote{Czech Technical University in Prague, The Czech Republic. E-mail: \texttt{bosansky@agents.fel.cvut.cz}}
\and
Simina Br\^{a}nzei\footnote{
Hebrew University of Jerusalem, Israel. E-mail: \texttt{simina.branzei@gmail.com}}
\and
Kristoffer Arnsfelt Hansen\footnote{Aarhus University, Denmark. E-mail: \texttt{arnsfelt@cs.au.dk}}
\and
Peter Bro Miltersen\footnote{Aarhus University, Denmark. E-mail: \texttt{bromille@cs.au.dk}}
\and
Troels Bjerre S{\o}rensen\footnote{IT-University of Copenhagen, Denmark. E-mail: \texttt{trbj@itu.dk}}
}
\date{}

\begin{document}

\maketitle

\begin{abstract}
The Stackelberg equilibrium is a solution concept that describes optimal strategies to commit
to: Player~1 (\emph{the leader}) first commits to a strategy that is publicly announced, 
then Player~2 (\emph{the follower}) plays a best response to the leader's choice.
We study the problem of computing Stackelberg equilibria in finite sequential (i.e., extensive-form) games and provide new exact algorithms, approximate
algorithms, and hardness results for finding equilibria for several classes of such two-player games.
\end{abstract}


\section{Introduction}
The Stackelberg competition is a game theoretic model introduced by von Stackelberg~\cite{Stackelberg34} for studying market structures. 
The original formulation of a Stackelberg duopoly captures the scenario of two firms that compete by selling homogeneous products. 
One firm---\emph{the leader}---first decides the quantity to sell and announces it publicly, while the second firm---\emph{the follower}---decides its own production only after observing the announcement of the first firm. 
The leader firm must have commitment power (e.g., is the monopoly in an industry) and cannot undo its publicly announced strategy, while the follower firm (e.g., a new competitor) plays a best response to the leader's chosen strategy. 

The Stackelberg competition has been an important model in economics ever since (see, e.g.,~\cite{sherali1984multiple,HAMILTON199029,AMIR19991,VANDAMME1999105,Matsumura2003,Etro2007}), while the solution concept of a \emph{Stackelberg equilibrium} has been studied in a rich body of literature in computer science, with a number of important 
real-world applications developed in the last decade~\cite{tambe2011}.
The Stackelberg equilibrium concept can be applied to any game with two players (e.g., in normal or extensive form) and stipulates that the leader first commits to a strategy, while the follower observes the leader's choice and best responds to it. The leader must have commitment power; in the context of firms, the act of moving first in an industry, such as by opening a shop, requires financial investment and is evidently a form of commitment. In other scenarios, the leader's commitment refers to ways of responding to future events, should certain situations be reached, and in such cases the leader must have a way of enforcing credible threats. The leader can always commit to a Nash equilibrium strategy, however it can often obtain a better payoff by choosing some other strategy profile. While there exist generalizations of Stackelberg equilibrium to multiple players, the real-world implementations to date have been derived from understanding the two player model, and for this reason we will focus on the two-player setting.

One of the notable applications using the conceptual framework of Stackelberg equilibrium has been the development of algorithms for protecting airports and ports in the United States (deployed so far in Boston, Los Angeles, New York). More recent ongoing work (see, e.g., \cite{Nguyen2015}), explores additional problems such as protecting wildlife, forest, and fisheries. The general task of defending valuable resources against attacks can be cast in the Stackelberg equilibrium model as follows. The role of the leader is taken by the defender (e.g., police forces), who commits to a strategy, such as the allocation of staff members to a patrolling schedule of locations to check. The role of the follower is played by a potential attacker, who monitors the empirical distribution (or even the entire schedule) of the strategy chosen by the defender,
and then best responds, by devising an optimal attack given this knowledge.
The crucial question is how to minimize the damage from potential threats, by computing an optimal schedule for the defender. Solving this problem in practice involves several nontrivial steps, such as estimating the payoffs of the participants for the resources involved (e.g., the attacker's reward for destroying a section of an airport) and computing the optimal strategy that the defender should commit to.
\\

In this paper, we are interested in the following fundamental question:
\begin{quote}
{\em
Given the description of a game in extensive form, compute the optimal strategy that the leader should commit to.
}
\end{quote}
We study this problem for multiple classes of two-player extensive-form games (EFGs) and
variants of the Stackelberg solution concept that differ in kinds of strategies to commit to, and provide both efficient algorithms
and computational hardness results. 
We emphasize the positive results in the main text of the submission and fully state technical hardness results in the appendix.

\subsection{Our Results}

The problem of computing a Stackelberg equilibrium in EFGs can be classified by the following parameters:
\begin{itemize}
\item \emph{Information}. Information captures how much a player knows about the opponent's moves (past and present). We study \emph{turn-based games} (TB), where for each state there is a unique player that can perform an action, and \emph{concurrent-move games} (CM), where the players act simultaneously in at least one state. 
\item \emph{Chance}. A game with chance nodes allows stochastic transitions between states; otherwise, the transitions are deterministic (made through actions of the players).
\item \emph{Graph}. We focus on \emph{trees} and \emph{directed acyclic graphs} (DAGs) as the main representations. Given such a graph, each node represents a different state in the game, while the edges represent the transitions between states. 
\item \emph{Strategies}. We study several major types of strategies that the leader can commit to, namely \emph{pure} (P), \emph{behavioral} (B), and \emph{correlated behavioral} (C).
\end{itemize}

\begin{table}[ht]
\caption{Overview of the computational complexity results containing both existing and new results provided by this paper (marked with *). Information column: TB stands for \emph{turn-based} and CM for \emph{concurrent moves}. 
Strategies: P stands for \emph{pure}, B for \emph{behavioral}, and C for \emph{correlated}. Finally, $|\calS|$ denotes the number of decision points in the game and $|\calZ|$ the number of terminal states.\label{tab:results}}{%
\begin{tabular}{|l||c|c|c|c|c|c|}
\hline  & \textsc{Information} & \textsc{Chance} & \textsc{Graph} & \textsc{Strategies} & \textsc{Complexity} & \textsc{Source}\\
\hline
\hline ~1.* & TB & \xmark & DAG & P & $O\left(|\calS| \cdot (|\calS| + |\calZ|)\right)$ & {\it Theorem~\ref{th:DAG_pure}}\\
\hline ~2. & TB & \xmark & Tree & B & $O\left(|\calS| \cdot |\calZ|^2\right)$ & \cite{Letchford2010} \\
\hline ~3.* & TB & \xmark & Tree & C & $O\left(|\calS| \cdot |\calZ|\right)$ & {\it Theorem~\ref{th:tree_corr}} \\
\hline ~4. & TB & \cmark & Tree & B & NP-hard & \cite{Letchford2010} \\
\hline ~5.* & TB & \cmark & Tree & P & FPTAS & {\it Theorem~\ref{app-pure}} \\
\hline ~6.* & TB & \cmark & Tree & B & FPTAS & {\it Theorem~\ref{app-behavior}} \\
\hline ~7.* & TB & \cmark & Tree & C & $O\left(|\calS| \cdot |\calZ|\right)$ & {\it Theorem~\ref{th:tree_corr_ch}} \\
\hline ~8.* & CM & \xmark & Tree & B & NP-hard & {\it Theorem~\ref{th:cm_tree}} \\
\hline ~9.* & CM & \cmark & Tree & C & polynomial & {\it Theorem~\ref{th:cm_tree:LP}} \\
\hline
\end{tabular}}
\end{table}

The results are summarized in Table~\ref{tab:results} and can be divided in three categories\footnote{We stated a theorem for NP-hardness for the correlated case on DAGs that was similar to the original theorem for behavioral strategies~\cite{letchford2013_thesis} in an earlier version of this paper. Due to an error, the theorem has not been correctly proven and the computational complexity for this case (i.e., computing optimal correlated strategies to commit to on DAGs) remains currently open.}.  \\

First, we design a more efficient algorithm for computing optimal strategies for turn-based games on DAGs. Compared to the previous state of the art (due to Letchford and Conitzer~\cite{Letchford2010},~\cite{letchford2013_thesis}), we reduce the complexity by a factor proportional to the number of terminal states (see row 1 in Table~\ref{tab:results}). \\

Second, we show that correlation often reduces the computational complexity of finding optimal strategies. In particular, we design several new polynomial time algorithms for computing the optimal correlated strategy to commit to for both turn-based and concurrent-move games (see rows 3, 7, 9). \\

Third, we study approximation algorithms for the NP-hard problems in this framework and provide fully polynomial time approximation schemes for finding pure and behavioral Stackelberg equilibria for turn-based games on trees with chance nodes (see rows 5, 6). We leave open the question of finding an approximation for concurrent-move games on trees without chance nodes (see row 8).

\subsection{Related Work}
There is a rich body of literature studying the problem of computing Stackelberg equilibria.
The computational complexity of the problem is known for one-shot games~\cite{Conitzer2006}, Bayesian games~\cite{Conitzer2006}, and selected subclasses of extensive-form games~\cite{Letchford2010} and infinite stochastic games~\cite{Letchford12,gupta2015,gupta2015-conf}.
Similarly, many practical algorithms are also known and typically based on solving multiple linear programs~\cite{Conitzer2006}, 
or mixed-integer linear programs for Bayesian~\cite{Paruchuri2008} and extensive-form games~\cite{bosansky2015-aaai-sse}.

For one-shot games, the problem of computing a Stackelberg equilibrium is polynomial~\cite{Conitzer2006} in contrast to the PPAD-completeness of a Nash equilibrium~\cite{Daskalakis2006,Chen2009}.
The situation changes in extensive-form games where Letchford and Conitzer showed~\cite{Letchford2010} that for many cases the problem is NP-hard, while it still remains PPAD-complete for a Nash equilibrium~\cite{Daskalakis2006b}.
More specifically, computing Stackelberg equilibria is polynomial only for:
\begin{itemize}
\item games with perfect information with no chance on DAGs where the leader commits to a pure strategy,
\item games with perfect information with no chance on trees.
\end{itemize} 
Introducing chance or imperfect information leads to NP-hardness.
However, several cases were unexplored by the existing work, namely extensive-form games with perfect information and concurrent moves. 
We address this subclass in this work.

The computational complexity can also change when the leader commits to correlated strategies.
This extension of the Stackelberg notion to correlated strategies appeared in several works~\cite{Conitzer2011,Letchford12,xu2015}.
Conitzer and Korzhyk~\cite{Conitzer2011} analyzed correlated strategies in one-shot games providing a single linear program for their computation.
Letchford~\emph{et~al.}~\cite{Letchford12} showed that the problem of finding optimal correlated strategies to commit to is NP-hard in infinite discounted stochastic games\footnote{More precisely, that work assumes that the correlated strategies can use a finite history.}.
Xu~\emph{et~al.}~\cite{xu2015} focused on using correlated strategies in a real-world security based scenario.

The detailed analysis of the impact when the leader can commit to correlated strategies has, however, not been investigated sufficiently in the existing work.
We address this extension and study the complexity for multiple subclasses of extensive-form games.
Our results show that for many cases the problem of computing Stackelberg equilibria in correlated strategies is polynomial compared to the NP-hardness in behavioral strategies.
Finally, these theoretical results have also practical algorithmic implications.
An algorithm that computes a Stackelberg equilibrium in correlated strategies can be used to compute a Stackelberg equilibrium in behavioral strategies allowing a significant speed-up in computation time~\cite{cermak2016-aaai}.

\section{Preliminaries}

We consider finite two-player sequential games.
Note that for every finite set $K$, $\Delta(K)$ denotes probability distributions over $K$ and $\calP(K)$ denotes the set of all subsets of $K$.

\begin{definition}[2-player sequential game]
A \emph{two-player sequential game} is given by a tuple $G = (\calN, \calS,\calZ,$ $\rho, \calA,$ $u, \calT, \calC)$, where:
\begin{itemize}
\item $\calN = \{1,2\}$ is a set of two players;
\item $\calS$ is a set of non-terminal states;
\item $\calZ$ is a set of terminal states;
\item $\rho : \calS \rightarrow \calP(\calN) \cup \{c\}$ is a function that defines which player(s) act in a given state, or whether the node is a chance node 
(case in which $\rho(s) = c$);
\item $\calA$ is a set of actions; we overload the notation to restrict the actions only for a single player as $\calA_i$ and for a single state as $\calA(s)$;
\item $\calT : \calS \times \prod_{i \in \rho(s)}\calA_i \rightarrow \{\calS \cup \calZ\}$ is a transition function between states depending on the actions taken by all the players that act in this state. Overloading notation, $\calT(s)$ also denotes the children of a state $s$: $\calT(s) = \{ s' \in \calS \cup \calZ \;|\; \exists a \in \calA(s) ;\; \calT(s,a) = s' \}$;
\item $\calC : \calA_c \rightarrow [0,1]$ are the chance probabilities on the edges outgoing from each chance node $s \in \calS$, such that
$\sum_{a \in \calA_c(s)}\calC(a) = 1$;  
\item Finally, $u_i : \calZ \rightarrow \mathbb{R}$ is the utility function for player $i \in \calN$.
\end{itemize}
\end{definition}

In this paper we study Stackelberg equilibria, thus player~1 will be referred to as the \emph{leader} and player~2 as the \emph{follower}.

We say that a game is \emph{turn-based} if there is a unique player acting in each state (formally, $|\rho(s)| = 1 \; \forall s \in \calS$) and with \emph{concurrent moves} if both players can act simultaneously in some state.
Moreover, the game is said to have \emph{no chance} 
if there exist no chance nodes; otherwise the game is \emph{with chance}.

A \emph{pure strategy} $\pi_i \in \Pi_i$ of a player $i \in \calN$ is an assignment of an action to play in each state of the game ($\pi_i : \calS \rightarrow \calA_i$).
A \emph{behavioral strategy} $\sigma_i \in \Sigma_i$ is a probability distribution over actions in each state $\sigma_i : \calA \rightarrow [0,1]$ such that $\forall s \in S, \forall i \in \rho(s) \; \sum_{a \in \calA_i(s)}\sigma_i(a) = 1$.

The expected utility of player $i$ given a pair of strategies $(\sigma_1, \sigma_2)$ is defined as follows:
\[
u_i(\sigma_1,\sigma_2) = \sum_{z \in \calZ}u_i(z)p_\sigma(z),
\]
where $p_\sigma(z)$ denotes the probability that leaf $z$ will be reached if both players follow the strategy from $\sigma$ and due to stochastic transitions corresponding to $\calC$.

A strategy $\sigma_i$ of player $i$ is said to represent a \emph{best response} to the opponent's strategy $\sigma_{-i}$ if $u_i(\sigma_i,\sigma_{-i}) \geq u_i(\sigma'_i,\sigma_{-i}) \; \forall \sigma'_i \in \Sigma_i$. 
Denote  by$\calBR(\sigma_{-i}) \subseteq \Pi_i$ the set of all the pure best responses of player $i$ to strategy $\sigma_{-i}$.
We can now introduce formally the Stackelberg Equilibrium solution concept:

\begin{definition}[Stackelberg Equilibrium]
A strategy profile $\sigma = (\sigma_1,\sigma_2)$ is a \emph{Stackelberg Equilibrium} if $\sigma_1$ is an optimal strategy of the leader given that the follower best-responds to its choice. Formally, a Stackelberg equilibrium in {\em pure} strategies is defined as
$$(\sigma_1,\sigma_2) = \argmax_{\sigma'_1 \in \Pi_1, \sigma'_2 \in \calBR(\sigma'_1)} u_1(\sigma'_1,\sigma'_2)$$
while a Stackelberg equilibrium in {\em behavioral} strategies is defined as
$$(\sigma_1,\sigma_2) = \argmax_{\sigma'_1 \in \Sigma_1, \sigma'_2 \in \calBR(\sigma'_1)} u_1(\sigma'_1,\sigma'_2)$$
\end{definition}

Next, we describe the notion of a Stackelberg equilibrium where the leader can commit to a correlated strategy in a sequential game. The concept was suggested and investigated by Letchford {\em et al.} \cite{Letchford12}, but no formal definition exists. Formalizing such a definition below, we observe that the definition is essentially the ``Stackelberg analogue'' of the notion of \emph{Extensive-Form Correlated Equilibria (EFCE)}, introduced by von Stengel and Forges~\cite{vonStengel08}. This parallel turns out to be technically relevant as well.
\begin{definition}[Stackelberg Extensive-Form Correlated Equilibrium]
A probability distribution $\phi$ on pure strategy profiles $\Pi$ is called a \emph{Stackelberg Extensive-Form Correlated Equilibrium (SEFCE)} if it maximizes the leader's utility (that is, $\phi = \argmax_{\phi' \in \Delta(\Pi)}u_1(\phi')$) subject to the constraint that
whenever the play reaches a state $s$ where the follower can act, the follower is recommended an action $a$ according to $\phi$ such that the follower cannot gain by unilaterally deviating from $a$ in state $s$ (and possibly in all succeeding states), given the posterior on the probability distribution of the strategy of the leader, defined by the actions taken by the leader so far.
\end{definition}

We give an example to illustrate both variants of the Stackelberg solution concept. 

\begin{example}
Consider the game in Figure~\ref{fig:gt}, where the follower moves first (in states $s_1, s_2$) and the leader second (in states $s_3, s_4$).
By committing to a behavioral strategy, the leader can gain utility $1$ in the optimal case -- leader commits to play \emph{left} in state $s_3$ and \emph{right} in $s_4$.
The follower will then prefer playing \emph{right} in $s_2$ and \emph{left} in $s_1$, reaching the leaf with utilities $(1,3)$.
Note that the leader cannot gain more by committing to strictly mixed behavioral strategies. 

Now, consider the case when the leader commits to correlated strategies.
We interpret the probability distribution over strategy profiles $\phi$ as signals send to the follower in each node where the follower acts, while the leader is committing to play with respect to $\phi$ and the signals sent to the follower.
This can be shown in node $s_2$, where the leader sends one of two signals to the follower, each with probability $0.5$. 
In the first case, the follower receives the signal to move \emph{left}, while the leader commits to play uniform strategy in $s_3$ and action \emph{left} in $s_4$ reaching the utility value $(2,1)$ if the follower plays according to the signal.
In the second case, the follower receives the signal to move \emph{right}, while the leader commits to play \emph{right} in $s_4$ and \emph{left} in $s_3$ reaching the utility value $(1,3)$ if the follower plays according to the signal.
By using this correlation, the leader is able to get the utility of $1.5$, while ensuring the utility of $2$ for the follower; hence, the follower will follow the only recommendation in node $s_1$ to play \emph{left}.

\begin{figure}[t]
\centering
\includegraphics[width=0.5\textwidth]{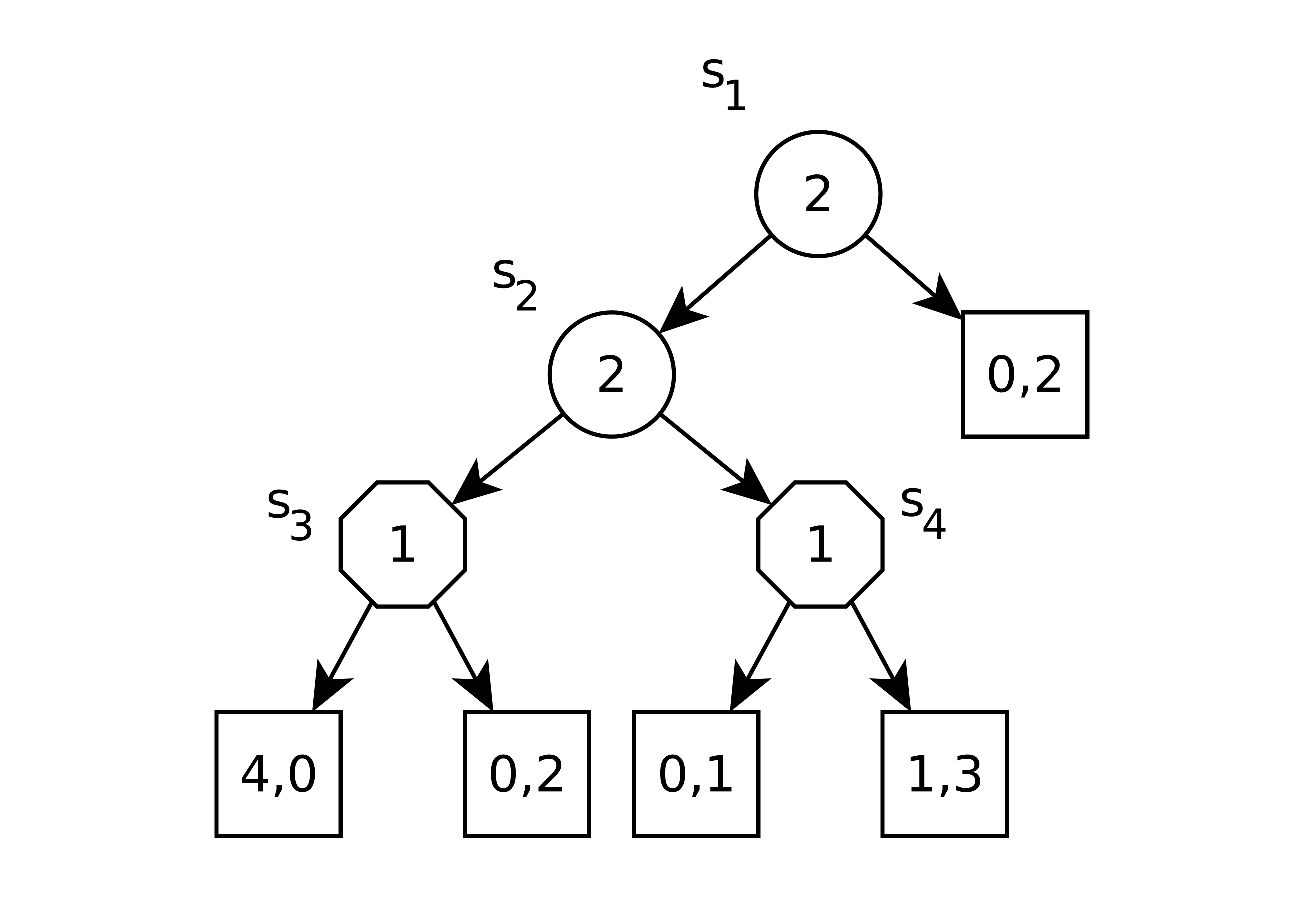}
\includegraphics[width=0.32\textwidth]{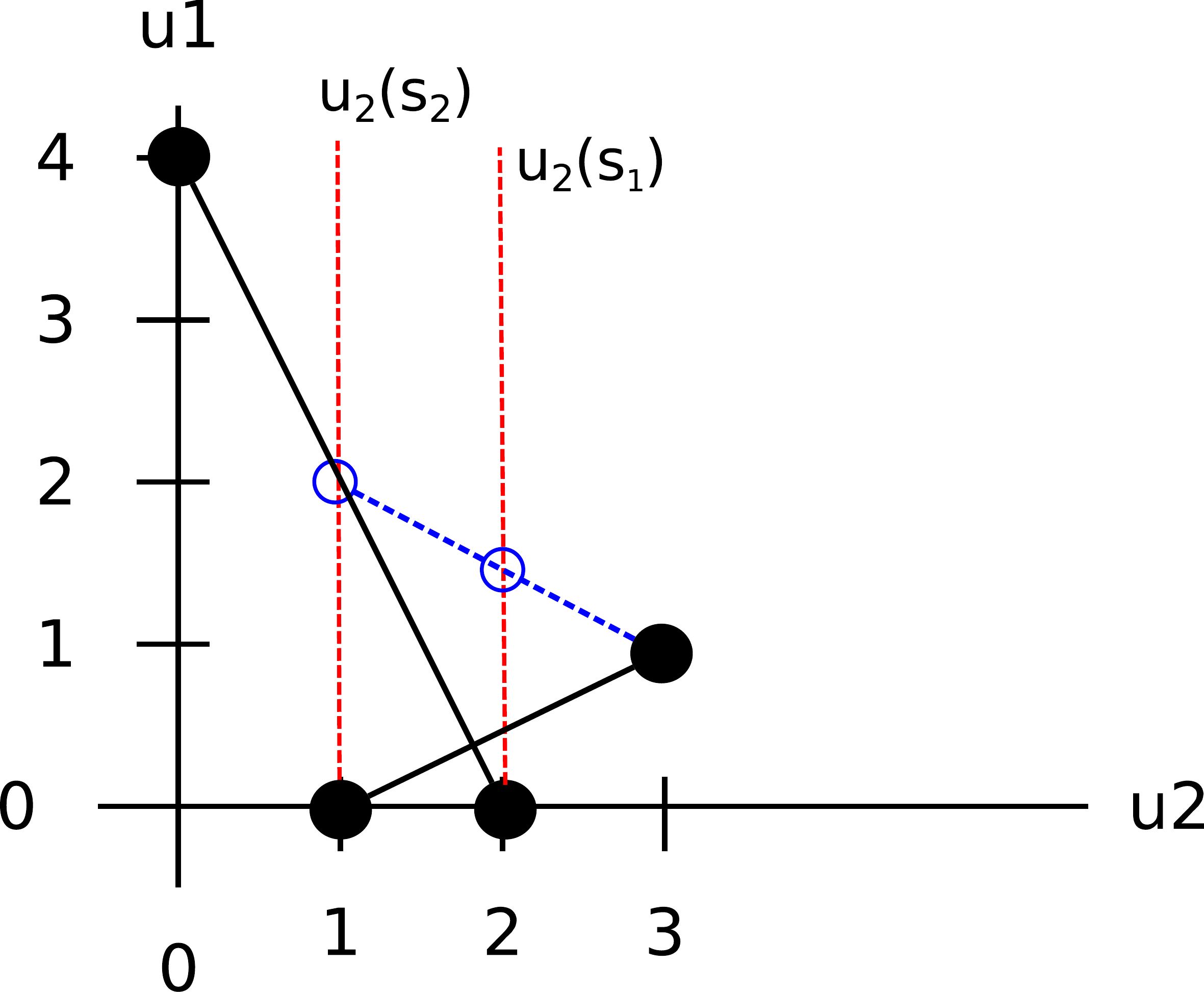}
\caption{
(Left) An example game with different outcomes depending on whether the leader commits to behavioral or to correlated strategies. The leader acts in nodes $s_3$ and $s_4$, the follower acts in nodes $s_1$ and $s_2$.
Utility values are shown in the terminal states, first value is the utility for the leader, second value is the utility of the follower.
(Right) A visualization of the outcomes of the example game in the two-dimensional utility space of the players -- the horizontal axis corresponds to the utility of the follower, the vertical axis corresponds to the utility of the leader. Red vertical lines visualize the minimal value the follower can guarantee in $s_2$ or $s_1$ respectively, blue lines and points correspond to new outcomes that can be achieved if the leader commits to correlated strategies
\label{fig:gt}}
\end{figure}

The situation can be visualized using a two-dimensional space, where the $x$-axis represents the utility of the follower and the $y$-axis represents the utility of the leader.
This type of visualization was also used in \cite{Letchford2010} and we use it further in the proof of Theorem~\ref{th:tree_corr}.
While the black nodes correspond to the utility points of the leafs, the solid black lines correspond to outcomes when the leader randomize between the leafs.
The follower plays a best-response action in each node; hence, in order to force the follower to play action \emph{left} in $s_2$, the leader must guarantee the follower the utility of at least $1$ in the sub-game rooted in node $s_3$ since the follower can get at least this value by playing \emph{right} in $s_2$.
Therefore, each state of the follower restricts the set of possible outcomes of the game. 
These restrictions are visualized as the vertical dashed lines -- one corresponds to the described situation in node $s_2$, and the second one due to the leaf following node $s_1$.
Considering only commitments to behavioral strategies, the best of all possible outcomes for the leader is the point $(u_2=3,u_1=1)$.
With correlation, however, the leader can achieve a mixture of points $(u_2=1,u_1=2)$ and $(u_2=3,u_1=1)$ (the blue dashed line).
This can also be interpreted as forming a convex hull over all possible outcomes in the sub-tree rooted in node $s_2$. 
Note, that without correlation, the set of all possible outcomes is not generally a convex set.
Finally, after restricting this set of possible solutions due to leaf in node $s_1$, the intersection point $(u_2=2,u_1=1.5)$ represents the expected utility for the Stackelberg Extensive-Form Correlated Equilibrium solution concept. 
\end{example}


The example gives an intuition about the structure of the probability distribution $\phi$ in SEFCE.
In each state of the follower, the leader sends a signal to the follower and commits to follow the correlated strategy if the follower admits the recommendation, while simultaneously committing to punish the follower for each deviation.
This punishment is simply a strategy that minimizes the follower's utility and will be useful in many proofs; next we introduce some notation for it.

Let $\sigma^m$ denote a behavioral strategy profile, where in each sub-game the leader plays a minmax behavior strategy based on the utilities of the follower and the follower plays a best response.
Moreover, for each state $s \in \calS$, we denote by $\mu(s)$ the expected utility of the follower in the sub-game rooted in state $s$ if both players play according to~$\sigma^m$ (i.e., the value of the corresponding zero-sum sub-game defined by the utilities of the follower).

Note that being a probability distribution over pure strategy profiles, a SEFCE is, a priori, an object of exponential size in the size of the description of the game, when it is described as a tree. This has to be dealt with before we can consider computing it. The following lemma gives a compact representation of the correlated strategies in a SEFCE and the proof yields an algorithm for constructing the probability distribution $\phi$ from the compact representation. It is this compact representation that we seek to compute.

\begin{lemma}\label{lem:compactSEFCE}
For any turn-based or concurrent-move game in tree form, there exists a SEFCE $\phi \in \Delta(\Pi)$ that can be compactly represented as a behavioral strategy profile $\sigma = (\sigma_1,\sigma_2)$ such that $\forall z \in \calZ \; p_{\phi}(z) = p_{\sigma}(z)$ and $\phi$ corresponds to the following behavior:
\begin{itemize}
\item the follower receives signals in each state $s$ according to $\sigma_2(a)$ for each action $a \in \calA_2(s)$
\item the leader chooses the action in each state $s$ according to $\sigma_1(a)$ for each action $a \in \calA_1(s)$ if the state $s$ was reached by following the recommendations
\item both players switch to the minmax strategy $\sigma^m$ after a deviation by the follower. 
\end{itemize}
\end{lemma}
\begin{proof}
Let $\phi'$ be a SEFCE. 
We construct the behavioral strategy profile $\sigma$ from $\phi'$ and then show how an optimal strategy $\phi$ can be constructed from $\sigma$ and $\sigma^m$.

To construct $\sigma$, it is sufficient to specify a probability $\sigma(a)$ for each action $a \in \calA(s)$ in each state $s$.
We use the probability of state $s$ being reached (denoted $\phi'(s)$) that corresponds to the sum of pure strategy profiles $\phi'(\pi)$ such that the actions in strategy profile $\pi$ allow state $s$ to be reached.

Formally, there exists a sequence $s_0,a_0,\ldots,a_{k-1},s_k$ of states and actions (starting at the root), such that for every $j = 0, \ldots, k-1$ it holds that $a_j = \pi(s_j)$, $s_{j+1} = \calT(s_j,a_j)$ (or $s_{j+1}$ is the next decision node of some player if $\calT(s_j,a_j)$ is a chance node), $s_0 = s_{root}$, and $s_k = s$.
Let $\Pi(s)$ denote a set of pure strategy profiles for which such a sequence exists for state $s$, and $\Pi(s,a) \subseteq \Pi(s)$ the strategy profiles that not only reach $s$, but also prescribe action $a$ to be played in state $s$. 
We have:
$$\sigma(a) = \frac{\sum_{\pi' \in \Pi(s,a)}\phi'(\pi')}{\phi'(s)}, \mbox{ where } \phi'(s) = \sum_{\pi' \in \Pi(s)}\phi'(\pi')$$
In case $\phi'(s) = 0$, we set the behavior strategy in $\sigma$ arbitrarily.

Next, we construct a strategy $\phi$ that corresponds to the desired behavior and show that it is indeed an optimal SEFCE strategy.
We need to specify a probability for every pure strategy profile $\pi = (\pi_1,\pi_2)$.
Consider the sequence of states and actions that corresponds to executing the actions from the strategy profile $\pi$.
Let $s^l_0,a^l_0,\ldots,a^l_{k_l-1},s_{k_l}$ be one of $q$ possible sequences of states and actions (there can be multiple such sequences due to chance nodes), such that $j = 0, \ldots, k_l-1$, $a^l_j = \pi(s^l_j)$, $s^l_{j+1} = \calT(s^l_j,a^l_j)$ (or $s^l_{j+1}$ is one of the next decision nodes of some player immediately following the chance node(s) $\calT(s^l_j,a^l_j)$), $s^l_0 = s_{root}$, and $s^l_{k_l} \in \calZ$.
The probability for the strategy profile $\pi$ corresponds to the probability of executing the sequences of actions multiplied by the probability that the remaining actions prescribe minmax strategy $\sigma^m$ in case the follower deviates:
\[
\phi(\pi) = \left(\prod_{l=1}^{q}\prod_{j=0}^{k_l-1}\sigma(a^l_j)\right) \cdot \prod_{a' = \pi(s') | s' \in \calS \setminus \{s^1_0,\ldots,s^1_{k_0-1},s^2_0,\ldots,s^q_{k_q-1}\}}\sigma^m(a').
\]

\paragraph{Correctness}
By construction of $\sigma$ and $\phi$, it holds that probability distribution over leafs remains the same as in $\phi'$; hence, $\forall z \in \calZ \; p_{\phi'}(z) = p_{\sigma}(z) = p_{\phi}(z)$ and thus the expected utility of $\phi$ for the players is the same as in $\phi'$.

Second, we have to show that the follower has no incentive to deviate from the recommendations in $\phi$.
By deviating to some action $a'$ in state $s$, the follower gains $\mu(\calT(s,a'))$, since both players play according to $\sigma^m$ after a deviation.
In $\phi'$, the follower can get for the same deviation at best some utility value $v_2(\calT(s,a'))$, which by the definition of the minmax strategies $\sigma^m$ is greater or equal than $\mu(\calT(s,a'))$.
Since the expected utility of the follower for following the recommendations is the same in $\phi$ as in $\phi'$, and the follower has no incentive to deviate in $\phi'$ because of the optimality, the follower has no incentive to deviate in $\phi$ either.
\end{proof}

\section{Computing Exact Strategies in Turn-Based Games}\label{sec:TB}
We start our computational investigation with turn-based games. 

\begin{theorem}\label{th:DAG_pure}
There is an algorithm that takes as input a turn-based game in DAG form with no chance nodes and outputs a Stackelberg equilibrium in pure strategies. The algorithm runs in time $O(|\calS|(|\calS| + |\calZ|))$.
\end{theorem}
\begin{proof}
Our algorithm performs three passes through all the nodes in the graph. 

First, the algorithm computes the minmax values $\mu(s)$ of the follower for each node in the game by backward induction.

Second, the algorithm computes a {\em capacity} for each state in order to determine which states of the game are reachable (i.e., there exists a commitment of the leader and a best response of the follower such that the state can be reached by following their strategies). 
The capacity of state~$s$, denoted $\gamma(s)$, is defined as the minimum utility of the follower that needs to be guaranteed by the outcome of the sub-game starting in state $s$ in order to make this state reachable. 
By convention $\gamma(s_{root}) = -\infty$ and we initially set $\gamma(\calS \cup \calZ \setminus \{s_{root}\}) = \infty$ and mark them as open.


Third, the algorithm evaluates each open state $s$, for which all parents have been marked as closed.
We distinguish whether the leader, or the follower makes the decision:
\begin{itemize}
\item \emph{$s$ is a leader node}: the algorithm sets $\gamma(s') = \min(\gamma(s'),\gamma(s))$ for all children $s' \in \calT(s)$;
\item \emph{$s$ is a follower node}: the algorithm sets $\gamma(s') = \min(\gamma(s'),\max(\gamma(s),\max_{s'' \in \calT(s) \setminus \{ s'\} }\mu(s'')))$ for all children $s' \in \calT(s)$.
\end{itemize}
Finally, we mark state $s$ as closed.

We say that leaf $z \in \calZ$ is a {\em possible outcome}, if $\mu(z) = u_2(z) \geq \gamma(z)$.
Now, the solution is such a possible outcome that maximizes the utility of the leader, i.e.
$
\argmax_{z \in \calZ \; u_2(z) \geq \gamma(z)}u_1(z)
$.
The strategy is now constructed by following nodes from leaf $z$ back to the root while using nodes $s'$ with capacities $\gamma(s') \leq \mu(z)$. 
Due to the construction of capacities, such a path exists and forms a part of the Stackelberg strategy.
The leader commits to the strategy leading to $\max\min$ utility for the follower in the remaining states that are not part of this path.

\paragraph{Complexity Analysis}
Computing the $\max\min$ values can be done in $O(|\calS|(|\calS| + |\calZ|))$ by backward induction due to the fact the graph is a DAG. 
In the second pass, the algorithm solves the widest-path problem from a single source to all leafs.
In each node, the algorithm calculates capacities for every child.
In nodes where the leader acts, there is a constant-time operation performed for each child.
However, we need to be more careful in nodes where the follower acts.
For each child $s' \in \calT(s)$ the algorithm computes a maximum value $\mu(s')$ of all of the siblings. 
We can do this efficiently by computing two maximal values of $\mu(s')$ for all $s' \in \calT(s)$ (say $s^1, s^2$) and for each child then the term $\max_{s'' \in \calT(s) \setminus \{ s'\} }\mu(s'')$ equals either to $s^1$ if $s' \neq s^1$, or to $s^2$ if $s' = s^1$.
Therefore, the second pass can again be done in $O(|\calS|(|\calS| + |\calZ|))$.
Finally, finding the optimal outcome and constructing the optimal strategy is again at most linear in the size of the graph.
Therefore the algorithm takes at most $O(|\calS|(|\calS| + |\calZ|))$ steps.
\end{proof} ~\\

Next we provide an algorithm for computing a Stackelberg extensive-form correlated equilibrium for turn-based games with no chance nodes.

\begin{theorem}\label{th:tree_corr}
There is an algorithm that takes as input a turn-based game in tree form with no chance nodes and outputs an SEFCE in the compact representation. The algorithm runs in time $O(|\calS||\calZ|)$.
\end{theorem}

\begin{proof}
We improve the algorithm from the proof of Theorem 4 in \cite{Letchford2010}.
The algorithm contains two steps: (1) a bottom-up dynamic program that for each node $s$ computes the set of possible outcomes, (2) a downward pass constructing the optimal correlated strategy in the compact representation.

For each node $s$ we keep set of points $H_s$ in two-dimensional space, where the $x$-dimension represents the utility of the follower and the $y$-dimension represents the utility of the leader.
These points define the convex set of all possible outcomes of the sub-game rooted in node $s$ (we assume that $H_s$ contains only the points on the boundary of the convex hull).
We keep each set $H_s$ sorted by polar angle.

\paragraph{Upward pass}
In leaf $z \in \calZ$, we set $H_z = \{z\}$.
In nodes $s$ where the leader acts, the set of points $H_s$ is equal to the convex hull of the corresponding sets of the children $H_w$. That is, $H_s = \mbox{\rm Conv}(\cup_{w \in {\cal T}(s)} H_w)$.

In nodes $s$ where the follower acts, the algorithm performs two steps.
First, the algorithm removes from each set $H_w$ of child $w$ the outcomes from which the follower has an incentive to deviate.
To do this, the algorithm uses the maxmin $u_2$ values of all other children of $s$ except $w$ and creates a new set $\hat{H}_w$ that we call the \emph{restricted set}.
The restricted set $\hat{H}_w$ is defined as an intersection of the convex set representing all possible outcomes $H_w$ and all outcomes defined by the halfspace restricting the utility $x$ of the follower by the inequality:
$$ x \geq \max_{w' \in \calT(s); w'\neq w}\;\min_{p' \in H_{w'}} u_2(p').$$

Second, the algorithm computes the set $H_s$ by creating a convex hull of the corresponding restricted sets $\hat{H}_w$ of the children $w$. That is,
$H_s = \mbox{\rm Conv}(\cup_{w \in \calT(s)} \hat{H}_w).$

Finally, in  the root of the game tree, the outcome of the Stackelberg Extensive-Form Correlated Equilibrium is the point with maximal payoff of player 1: $ p_{SE} = \argmax_{p \in H_{s_{root}}} u_1(p).$

\paragraph{Downward pass}
We now construct the compact representation of commitment to correlated strategies that ensures the outcome $p_{SE}$ calculated in the upward pass.
The method for determining the optimal strategy in each node is similar to the method $\textsf{strategy}(s,p'')$ used in the proof of Theorem 4 in \cite{Letchford2010}.

Given a node $s$ and a point $p''$ that lies on the boundary of $H_s$, this method specifies how to commit to correlated strategies in the sub-tree rooted in node $s$.
Moreover, the proof in~\cite{Letchford2010} also showed that it is sufficient to consider mixtures of at most two actions in each node and allowing correlated strategies does violate their proof. We consider separately leader and follower nodes: \\

$\bullet$ \hspace{0.2mm} For each node $s$ where the leader acts, the algorithm needs to find two points $p,p'$ in the boundaries of children $H_{w}$ and $H_{w'}$, such that the desired point $p''$ is a convex combination of $p \in H_{w}$ and $p' \in H_{w'}$. 
If $w = w'$, then the strategy in node $s$ is to commit to pure strategy leading to node $w$.
If $w \neq w'$, then the strategy to commit to in node $s$ is a mixture: with probability $\alpha$ to play action leading to $w$ and with probability $(1-\alpha)$ to play action leading to $w'$, where $\alpha \in [0,1]$ is such that $p'' = \alpha p + (1-\alpha)p'$.
Finally, for every child $s' \in \calT(s)$ we call the method $\textsf{strategy}$ with appropriate $p$ (or $p'$) in case $s' = w$ (or $w'$), and with the threat value corresponding to $\mu(s')$ for every other child. \\

$\bullet$ \hspace{0.2mm}
 For each node $s$ where the follower acts, the algorithm again needs to find two points $p,p'$ in the \emph{restricted} boundaries of children $\hat{H}_{w}$ and $\hat{H}_{w'}$, such that the desired point $p''$ is a convex combination of $p \in \hat{H}_{w}$ and $p' \in \hat{H}_{w'}$.
The reason for using the restricted sets is because the follower must not have an incentive to deviate from the recommendation.

Similarly to the previous case, if $w = w'$, then the correlated strategy in node $s$ is to send the follower signal leading to node $w$ while committing further to play $\textsf{strategy}(w,p)$ in sub-tree rooted in node $w$, and to play the minmax strategy in every other child $s'$ corresponding to value $\mu(s')$.

If $w \neq w'$, then there is a mixture of possible signals: with probability $\alpha$ the follower receives a signal to play the action leading to $w$ and with probability $(1-\alpha)$ signal to play the action leading to $w'$, where $\alpha \in [0,1]$ is again such that $p'' = \alpha p + (1-\alpha)p'$.
As before, by sending the signal to play certain action, the leader commits to play method $\textsf{strategy}(w,p)$ (or $\textsf{strategy}(w',p')$) in sub-tree rooted in node $w$ (or $w'$) and committing to play the minmax strategy leading to value $\mu(s')$ for every other child $s'$.

\paragraph{Correctness}
Due to the construction of the set of points $H_s$ that are maintained for each node $s$, these points correspond to the convex hull of all possible outcomes in the sub-game rooted in node $s$.
In leafs, the algorithm adds the point corresponding to the leaf.
In the leader's nodes, the algorithm creates a convex combinations of all possible outcomes in the children of the node.
The only places where the algorithm removes some outcomes from these sets are nodes of the follower.
If a point is removed from $H_w$ in node $s$, there exists an action of the follower in $s$ that guarantees the follower a strictly better expected payoff than the expected payoff of the outcome that correspond to the removed point.
Therefore, such an outcome is not possible as the follower will have an incentive to deviate.
The outcome selected in the root node is the possible outcome that maximizes the payoff of the leader of all possible outcomes; hence, it is optimal for the leader.
Finally, the downward pass constructs the compact representation of the optimal correlated strategy to commit to that reaches the optimal outcome.

\paragraph{Complexity Analysis}
Computing boundary of the convex hull $H_s$ takes $O(|\calZ|)$ time in each level of the game tree since the children sets $H_w$ are already sorted~\cite[p.~6]{MarkyMark2000}.
Moreover, since we keep only nodes on the boundary of the convex hull, the inequality $\sum_{s \in \calS} |H_s| \leq |\calZ|$ for all nodes in a single level of the game tree also bounds the number of lines that need to be checked in the downward pass.
Therefore, each pass takes at most $O(|\calS||\calZ|)$ time. 
\end{proof}

Interestingly, the algorithm described in the proof of Theorem~\ref{th:tree_corr} can be modified also in cases where the game contains chance, as shown in the next theorem.
This is in contrast to computing a Stackelberg equilibria that is NP-hard with chance.

\begin{theorem}\label{th:tree_corr_ch}
There is an algorithm that takes as input a turn-based game in tree form with chance nodes and outputs the compact form of an SEFCE for the game. The algorithm runs in time $O(|\calS||\calZ|)$.
\end{theorem}

\begin{proof}
We can use the proof from Theorem~\ref{th:tree_corr}, but need to analyze what happens in chance nodes in the upward pass.
The algorithm computes in chance nodes the Minkowski sum of all convex sets in child nodes and since all sets are sorted and this is a planar case, this operation can be again performed in linear time~\cite[p.~279]{MarkyMark2000}.
The size of set $H_s$ is again bounded by the number of all leafs~\cite{gritzmann93}. 
\end{proof}


\section{Computing Exact Strategies in Concurrent-Move Games}\label{sec:CM}

Next we analyze concurrent-move games and show that while the problem of computing a Stackelberg equilibrium in behavior strategies is NP-hard (even without chance nodes), the problem of computing a Stackelberg extensive-form correlated equilibrium can be solved in polynomial time.

\begin{theorem}\label{th:cm_tree}
Given a concurrent-move games in tree form with no chance nodes and a number $\alpha$,
 it is NP-hard to decide if the leader achieves payoff at least $\alpha$ in a Stackelberg equilibrium in behavior strategies.
\end{theorem}
The proof for the above hardness result above is included in the appendix Section~\ref{sec:a2}; the proof uses a reduction from the NP-complete problem \textsc{Knapsack}.

\begin{theorem} \label{th:cm_tree:LP}
For a concurrent-move games in tree form, the compact form of an SEFCE for the game can be found in polynomial time by solving a single linear program.
\end{theorem}
\begin{proof}
We construct a linear program (LP) based on the LP for computing Extensive-Form Correlated Equilibria (EFCE)~\cite{vonStengel08}.
We use the compact representation of SEFCE strategies (described by Lemma~\ref{lem:compactSEFCE}) represented by variables $\delta(s)$ that denote a joint probability that state $s$ is reached when both players, and chance, play according to SEFCE strategies.

The size of the original EFCE LP---both the number of variables and constraints---is quadratic in the number of sequences of players.
However, the LP for EFCE is defined for a more general class of imperfect-information games without chance.
In our case, we can exploit the specific structure of a concurrent-move game and together with the Stackelberg assumption reduce the number of constraints and variables.

First, the deviation from a recommended strategy causes the game to reach a different sub-game in which the strategy of the leader can be chosen (almost) independently to the sub-game that follows the recommendation. 

Second, the strategy that the leader should play according to the deviations is a minmax strategy, with which the leader punishes the follower by minimizing the utility of the follower as much as possible.
Thus, by deviating to action $a'$ in state $s$, the follower can get at best the minmax value of the sub-game starting in node $\calT(s,a')$ that we denote as $\mu(\calT(s,a'))$.
The values $\mu(s)$ for each state $s \in \calS$ can be computed beforehand using backward induction. \\

The linear program is as follows.
\begin{small}
\begin{eqnarray}
& & \max_{\delta,v_2}  \; \; \sum_{z \in \calZ}\delta(z)u_1(z) \qquad\\
\mbox{subject to:} \; \; \; \; \; \; \; \; \; \; \delta(s_{root}) & = & 1 \label{eq:cmlp:root}\\
0 \geq \delta(s) & \geq & 1 \qquad\qquad\qquad\;\; \forall s\in \calS\ \label{eq:cmlp:nf1}\\
\delta(s) & = & \sum_{s' \in \calT(s)}\delta(s') \qquad \forall s\in \calS;\; \rho(s) = \{1,2\} \label{eq:cmlp:nf2}\\
\delta(\calT(s,a_c)) & = & \delta(s)\calC(s,a_c) \qquad \forall s\in \calS\;\forall a\in \calA_c(s); \rho(s) = \{c\} \label{eq:cmlp:nf3}\\
v_2(z) & = & u_2(z)\delta(z) \qquad\quad \forall z \in \calZ \label{eq:cmlp:v1}\\
v_2(s) & = & \sum_{s' \in \calT(s)}v_2(s') \quad\;\; \forall s \in \calS \label{eq:cmlp:v2}\\
\sum_{a_1 \in \calA_1(s)}v_2(\calT(s,a_1\times a_2)) & \geq & \sum_{a_1 \in \calA_1(s)}\delta(\calT(s,a_1\times a_2))\mu(\calT(s,a_1\times a'_2)) \nonumber \\
& & \qquad\qquad\qquad\quad \forall s \in \calS \; \forall a_2,a_2' \in \calA_2(s) \label{eq:cmlp:br}
\end{eqnarray}
\end{small}
The interpretation is as follows. Variables $\delta$ represent the compact form of the correlated strategies. 

Equation~(\ref{eq:cmlp:root}) ensures that the probability of reaching the root state is~$1$, while Equation~(\ref{eq:cmlp:nf1}) ensures that for each state $s$, we have $\delta(s)$ between $0$ and $1$. \\

\emph{Network-flow constraints}: the probability of reaching a state equals the sum of probabilities of reaching all possible children (Equation~(\ref{eq:cmlp:nf2})) and it must correspond with the probability of actions in chance nodes (Equation~(\ref{eq:cmlp:nf3})).
The objective function ensures that the LP finds a correlated strategy that maximizes the leader's utility. \\

\emph{The follower has no incentive to deviate from the recommendations given by $\delta$}:
To this end, variables $v_2(s)$ represent the expected payoff for the follower in a sub-game rooted in node $s \in \calS$ when played according to $\delta$; defined by Equations~(\ref{eq:cmlp:v1}-\ref{eq:cmlp:v2}).
Each action that is recommended by $\delta$ must guarantee the follower at least the utility she gets by deviating from the recommendation.
This is ensured by Equation~(\ref{eq:cmlp:br}), where the expected utility for recommended action $a_2$ is expressed by the left side of the constraint, while the expected utility for deviating is expressed by the right side of the constraint.

Note that the expected utility on the right hand side of Equation~(\ref{eq:cmlp:br}) is calculated by considering the posterior probability after receiving the recommendation $a_2$ and the minmax values of children states after playing $a_2'$; $\mu(\calT(s,a_1\times a'_2))$.

Therefore, the variables $\delta$ found by solving this linear program correspond to the compact representation of the optimal SEFCE strategy.
\end{proof}

\section{Approximating Optimal Strategies}\label{sec:fptas}
In this section, we describe fully polynomial time approximation schemes for
finding a Stackelberg equilibrium in behavioral strategies as well as in pure strategies for turn based games on trees with chance
nodes. 

We start with the problem computing behavioral strategies for turn-based games on trees with chance nodes.

\begin{theorem}\label{app-behavior}
	There is an algorithm that takes as input a turn-based game on a tree
	with chance nodes and a parameter $\epsilon$, and computes a behavioral strategy for the leader. That strategy, combined with some best response of the follower, achieves a payoff that differs by at most $\epsilon$ from the payoff of the leader in a Stackelberg equilibrium in behavioral strategies. The algorithm runs in time $O(\epsilon^{-3}(UH_T)^3T)$, where
	$U=\max_{\sigma,\sigma'}u_1(\sigma)-u_1(\sigma')$, $T$ is the size of the
	game tree and $H_T$ is its height.
\end{theorem}
\begin{proof}
The exact version of this problem was shown to be NP-hard by Letchford and
Conitzer~\cite{Letchford2010}. Their hardness proof was a reduction from
{\sc Knapsack} and our algorithm is closely related to the classical approximation
scheme for this problem. We present here the algorithm, and delegate the proof of correctness to the appendix. 

Our scheme uses dynamic programming to construct a table of values for each
node in the tree. Each table contains a discretized representation of the
possible tradeoffs between the
utility that the leader can get and the utility that can at the same time be 
offered to the follower. In the appendix, we show that the cumulative error in the leaders 
utility is bounded additively by the height of the tree. This error only depends
on the height of the tree and not the utility. By an initial scaling of the
leader utility by a factor $D$, the error can be made arbitrarily small, at the
cost of extra computation time. This
scaling is equivalent to discretizing the leaders payoff to multiples of some
small $\delta=1/D$. For simplicity, we only describe the scheme for binary trees,
since nodes with higher branching factor can be replaced by small equivalent
binary trees.

An important property is that only the leader's utility is discretized, since we
need to be able to reason correctly about the follower's actions. The tables
are indexed by the leader's utility and contains values that are the follower's
utility. More formally, for each sub-tree $T$ we will compute a table $A_T$ with the
following guarantee for each index $k$ in each table:
\begin{enumerate}[a)]
\item the leader has a strategy for the game tree $T$ that offers the
follower utility $A_T[k]$ while securing utility {\em at least} $k$ to
the leader.
\item no strategy of the leader can (starting from sub-tree $T$) offer the
follower utility {\em strictly more} than $A_T[k]$, while securing 
utility {\em at least} $k+H_T$ to the leader, where $H_T$ is the 
height of the tree $T$.
\end{enumerate}
This also serves as our induction hypothesis for proving correctness. For
commitment to pure strategies, a similar table is used with the same guarantee,
except quantifying over pure strategies instead.

We will now examine each type of node, and for each show how the table is 
constructed. For each node $T$, we let $L$ and $R$ denote the two successors (if
any), and we let $A_T$, $A_L$, and $A_R$ denote their respective tables. Each
table will have $n=H_TU/\epsilon$ entries.

\paragraph{\bf If $T$ is a leaf} with utility $(u_1,u_2)$, the table can be
filled directly from the definition:
\begin{align*}
A_T[k] &:= \left\{\begin{matrix}u_2&\text{, if }k\leq u_1\\-\infty&\text{, otherwise}\end{matrix}\right.
\end{align*}
Both parts of the induction hypothesis are trivially satisfied by this.

\paragraph{\bf If $T$ is a leader node,} and the leader plays $L$ with
probability $p$, followed up by the strategies that gave the guarantees for
$A_L[i]$ and $A_R[j]$, then the leader would get an expected
$pi+(1-p)j$, while being able to offer $p A_L[i]+(1-p) A_R[j]$ to the follower.
For a given $k$, the optimal combination of the computed tradeoffs becomes:
	$A_T[k] := \max_{i,j,p} \{ p A_L[i]+(1-p) A_R[j]\quad|\quad p i+(1-p) j\geq k\}$.
This table can be computed in time $O(n^3)$ by looping over all $0\leq i,j,k<n$,
and taking the maximum with the extremal feasible values of $p$.

\paragraph{\bf If $T$ is a chance node,} where the probability of $L$ is $p$,
and the leader combines the strategies that gave the guarantees for $A_L[i]$ and
$A_R[j]$, then the leader would get an expected $pi+(1-p)j$ while being able to
offer $p A_L[i]+(1-p) A_R[j]$ to the follower.
For a given $k$, the optimal combination of the computed tradeoffs becomes:
$	A_T[k] := \max_{i,j} \{ p A_L[i]+(1-p) A_R[j]\quad|\quad p i+(1-p) j\geq k\}$.
The
table $A_T$ can thus be filled in time $O(n^3)$ by looping over all $0\leq
i,j,k<n$, and this can even be improved to $O(n^2)$ by a simple optimization.

\paragraph{\bf If $T$ is a follower node,} then if the leader combines the
strategy for $A_L[i]$ in $L$ with the minmax strategy for $R$, then the
followers best response is $L$ iff $A_L[i]\geq\mu(R)$, and similarly it is $R$
if $A_R[j]\geq\mu(L)$. Thus, the optimal combination becomes
\begin{align*}
	A_T[k] &:= \max(A_L[k]\downarrow_{\mu(R)},
	A_R[k]\downarrow_{\mu(L)})&x\downarrow_{\mu}:=\left\{\begin{matrix}x\text{, if 
}x\geq\mu\\-\infty\text{, otherwise}\end{matrix}\right.
\end{align*}
The table $A_T$ can be filled in time $O(n)$.

Putting it all together, each table can be computed in time $O(n^3)$, and there is one
table for each node in the tree, which gives the desired running time. Let $A_T$
be the table for the root node, and let $i'=\max\{i\ |\	A_T[i]>-\infty\}$. The
strategy associated with $A_T[i']$ guarantees utility that is at most $H_T$ from
the best possible guarantee in the scaled game, and therefore at most $\epsilon$
from the best possible guarantee in the original game.

This completes the proof of the theorem.
\end{proof} ~\\

Next, we prove the analogous statement for the case of pure strategies. Again, the exact problem was shown to be NP-hard by Conitzer and Letchford.

\begin{theorem}\label{app-pure}
	There is an algorithm that takes as input a turn-based game on a tree
	with chance nodes and a parameter $\epsilon$, and computes a pure strategy
	for the leader. That strategy, combined with some best response of the
	follower, achieves a payoff that differs by at most $\epsilon$ from the
	payoff of the leader in a Stackelberg equilibrium in pure strategies. The
	algorithm runs in time $O(\epsilon^{-2}(UH_T)^2T)$, where
	$U=\max_{\sigma,\sigma'}u_1(\sigma)-u_1(\sigma')$, $T$ is the size of the
	game tree and $H_T$ is its height.
\end{theorem}

\begin{proof}
The algorithm is essentially the same as the one for behavioral strategies, except that leader
nodes only have $p\in\{0,1\}$. The induction hypothesis is similar, except
the quantifications are over pure strategies instead. For a given $k$, the
optimal combination of the computed tradeoffs becomes:
	$$A_T[k] := \max\{A_c[i]\ |\ i\geq k\land c\in\{L,R\}\}.$$
The table $A_T$ can be computed in time $O(n)$.

The performance of the algorithm is slightly better than in the behavioral case, since the
most expensive type of node in the behavioral case can now be handled in linear time. Thus, computing each table now takes at
most $O(n^2)$ time, which gives the desired running time.
\end{proof}

\section{Discussion}
Our paper settles several open questions in the problem of complexity of computing a Stackelberg equilibrium in finite sequential games.
Very often the problem is NP-hard for many subclasses of extensive-form games and we show that the hardness holds also for games in the tree form with concurrent moves.
However, there are important subclasses that admit either an efficient polynomial algorithm, or fully polynomial-time approximation schemes (FPTAS); we provide an FPTAS for games on trees with chance.
The question unanswered within the scope of the paper is whether there exists a (fully) polynomial-time approximation scheme for games in the tree form with concurrent moves.
Our conjecture is that the answer is negative.

Second, we formalize a Stackelberg variant of the Extensive-Form Correlated Equilibrium solution concept (SEFCE) where the leader commits to correlated strategies.
We show that the complexity of the problem is often reduced (to polynomial) compared to NP-hardness when the leader commits to behavioral strategies.
However, this does not hold in general, which is showed by our hardness result for games on DAGs.

Our paper does not address many other variants of computing a Stackelberg equilibrium where the leader commits to correlated strategies.
First of all, we consider only two-player games with one leader and one follower.
Even though computing an Extensive-Form Correlated Equilibrium in games with multiple players is solvable in polynomial time, a recent result showed that computing a SEFCE on trees with no chance with 3 or more players is NP-hard~\cite{cerny2016}.
Second, we consider only behavioral strategies (or memoryless strategies) in games on DAGs.
Extending the concept of SEFCE to strategies that can use some fixed-size memory is a natural continuation of the present work.

\bibliographystyle{plain}

\begin{thebibliography}{10}

\bibitem{AMIR19991}
Rabah Amir and Isabel Grilo.
\newblock Stackelberg versus cournot equilibrium.
\newblock {\em Games and Economic Behavior}, 26(1):1 -- 21, 1999.

\bibitem{bosansky2015-aaai-sse}
Branislav Bosansky and Jiri Cermak.
\newblock {Sequence-Form Algorithm for Computing Stackelberg Equilibria in
  Extensive-Form Games}.
\newblock In {\em Proceedings of AAAI}, 2015.

\bibitem{cermak2016-aaai}
Jiri Cermak, Branislav Bosansky, Karel Durkota, Viliam Lisy, and Christopher
  Kiekintveld.
\newblock {Using Correlated Strategies for Computing Stackelberg Equilibria in
  Extensive-Form Games}.
\newblock In {\em Proceedings of AAAI Conference on Artificial Intelligence},
  2016.

\bibitem{cerny2016}
Jakub Cerny.
\newblock {Stackelberg Extensive-Form Correlated Equilibrium with Multiple
  Followers}.
\newblock Master's thesis, Czech Technical University in Prague, 2016.

\bibitem{Chen2009}
Xi~Chen, Xiaotie Deng, and Shang-Hua Teng.
\newblock {Settling the complexity of computing two-player Nash equilibria}.
\newblock {\em J. ACM}, 56:14:1--14:57, May 2009.

\bibitem{Conitzer2011}
Vincent Conitzer and Dmytro Korzhyk.
\newblock {Commitment to Correlated Strategies}.
\newblock In {\em Proceedings of AAAI}, pages 632--637, 2011.

\bibitem{Conitzer2006}
Vincent Conitzer and Tuomas Sandholm.
\newblock Computing the optimal strategy to commit to.
\newblock In {\em Proceedings of ACM-EC}, pages 82--90, 2006.

\bibitem{Daskalakis2006b}
Constantinos Daskalakis, Alex Fabrikant, and Christos~H. Papadimitriou.
\newblock {The Game World Is Flat: The Complexity of Nash Equilibria in
  Succinct Games}.
\newblock In {\em ICALP}, pages 513--524, 2006.

\bibitem{Daskalakis2006}
Constantinos Daskalakis, Paul~W. Goldberg, and Christos~H. Papadimitriou.
\newblock {The Complexity of Computing a Nash Equilibrium}.
\newblock In {\em Proceedings of the 38th annual ACM symposium on Theory of
  computing}, 2006.

\bibitem{MarkyMark2000}
Mark De~Berg, Marc Van~Kreveld, Mark Overmars, and Otfried~Cheong Schwarzkopf.
\newblock {\em Computational geometry}.
\newblock Springer, 2nd edition, 2000.

\bibitem{Etro2007}
Federico Etro.
\newblock {\em {Stackelberg Competition and Endogenous Entry}}, pages 91--129.
\newblock Springer Berlin Heidelberg, 2007.

\bibitem{gritzmann93}
Peter Gritzmann and Bernd Sturmfels.
\newblock {Minkowski addition of polytopes: computational complexity and
  applications to GrÃ¶bner bases}.
\newblock {\em SIAM J Discrete Math}, 6(2):246--269, 1993.

\bibitem{gupta2015}
Anshul Gupta.
\newblock {\em Equilibria in Finite Games}.
\newblock PhD thesis, University of Liverpool, 2015.

\bibitem{gupta2015-conf}
Anshul Gupta, Sven Schewe, and Dominik Wojtczak.
\newblock {Making the best of limited memory in multi-player discounted sum
  games.}
\newblock In {\em {Proceedings 6th International Symposium on Games, Automata,
  Logics and Formal Verification}}, pages 16--30, 2015.

\bibitem{HAMILTON199029}
Jonathan~H Hamilton and Steven~M Slutsky.
\newblock Endogenous timing in duopoly games: Stackelberg or cournot
  equilibria.
\newblock {\em Games and Economic Behavior}, 2(1):29 -- 46, 1990.

\bibitem{letchford2013_thesis}
Joshua Letchford.
\newblock {\em {Computational Aspects of Stackelberg Games}}.
\newblock PhD thesis, Duke University, 2013.

\bibitem{Letchford2010}
Joshua Letchford and Vincent Conitzer.
\newblock Computing optimal strategies to commit to in extensive-form games.
\newblock In {\em Proceedings of ACM-EC}, pages 83--92. ACM, 2010.

\bibitem{Letchford12}
Joshua Letchford, Liam MacDermed, Vincent Conitzer, Ronald Parr, and Charles~L.
  Isbell.
\newblock Computing optimal strategies to commit to in stochastic games.
\newblock In {\em Proceedings of AAAI}, pages 1380--1386, 2012.

\bibitem{Matsumura2003}
T.~Matsumura.
\newblock Stackelberg mixed duopoly with a foreign competitor.
\newblock {\em Bulletin of Economic Research}, 55:275â€“287, 2003.

\bibitem{Nguyen2015}
Thanh~H. Nguyen, Francesco M.~Delle Fave, Debarun Kar, Aravind~S.
  Lakshminarayanan, Amulya Yadav, Milind Tambe, Noa Agmon, Andrew~J. Plumptre,
  Margaret Driciru, Fred Wanyama, and Aggrey Rwetsiba.
\newblock {\em Making the Most of Our Regrets: Regret-Based Solutions to Handle
  Payoff Uncertainty and Elicitation in Green Security Games}, pages 170--191.
\newblock Springer International Publishing, Cham, 2015.

\bibitem{Paruchuri2008}
Praveen Paruchuri, J.P. Pearce, Janusz Marecki, M.~Tambe, F.~Ordonez, and Sarit
  Kraus.
\newblock {Playing games for security: an efficient exact algorithm for solving
  Bayesian Stackelberg games}.
\newblock In {\em Proceedings of AAMAS}, pages 895--902, 2008.

\bibitem{sherali1984multiple}
Hanif~D Sherali.
\newblock A multiple leader stackelberg model and analysis.
\newblock {\em Operations Research}, 32(2):390--404, 1984.

\bibitem{tambe2011}
Milind Tambe.
\newblock {\em {Security and Game Theory: Algorithms, Deployed Systems, Lessons
  Learned}}.
\newblock Cambridge University Press, 2011.

\bibitem{VANDAMME1999105}
Eric van Damme and Sjaak Hurkens.
\newblock Endogenous stackelberg leadership.
\newblock {\em Games and Economic Behavior}, 28(1):105 -- 129, 1999.

\bibitem{Stackelberg34}
H.~von Stackelberg.
\newblock Marktform und gleichgewicht.
\newblock {\em Springer-Verlag}, 1934.

\bibitem{vonStengel08}
B.~von Stengel and F.~Forges.
\newblock {Extensive-Form Correlated Equilibrium: Definition and Computational
  Complexity}.
\newblock {\em {Math Oper Res}}, 33(4):1002--1022, 2008.

\bibitem{xu2015}
Haifeng Xu, Zinovi Rabinovich, Shaddin Dughmi, and Milind Tambe.
\newblock {Exploring Information Asymmetry in Two-Stage Security Games}.
\newblock In {\em Proceedings of AAAI}, 2015.

\end{thebibliography}


\medskip
\newpage

\section{Appendix: Hardness Results}

In this section we provide the missing proof of NP-hardness.

\subsection{Computing Exact Strategies in Concurrent-Move Games}\label{sec:a2}

For the analysis in this section we use a variant of the NP-complete problem
\textsc{Knapsack}, which we call \textsc{Knapsack with
unit-items}:
\begin{quote}
\textsc{Knapsack with unit-items}:\emph{Given $N$ items with positive
  integer weights $w_1,\dots,w_N$ and values $v_1,\dots,v_N$, a weight
  budget $W$, and a target value $K$, and such that at least $W$ of the
  items have weight and value 1, does there exist $J \in \calP(N)$
  such that $\sum_{i \in J} w_i \leq W$ and $\sum_{i \in J} v_i \geq
  K$?}
\end{quote}

The following lemma will be useful.

\begin{lemma}
  The \textsc{Knapsack with
unit-items} problem is NP-complete.
\end{lemma}
\begin{proof}
  We can reduce from the ordinary \textsc{Knapsack} problem. So given
  $N$ items with weights $w_1,\dots,w_N$ and values $v_1,\dots,v_N$,
  and weight budget $W$ and target $K$, we form $N+W$ items. The
  weight and values of the first $N$ items are given by $w_i$ and
  $(W+1)v_i$, for $i=1,\dots,N$. The next $W$ items are given weight
  and value 1. The weight budget is unchanged $W$, but the new target
  value is $(W+1)K$.
\end{proof}

We can now prove the main result of this section.\\

\emph{\textsc{Theorem}~\ref{th:cm_tree} (restated). 
Given a concurrent-move games in tree form with no chance nodes and a number $\alpha$, it is NP-hard to decide if the leader achieves payoff at least $\alpha$ in a Stackelberg equilibrium in behavior strategies.}
\begin{proof}
Consider an instance of \textsc{Knapsack with unit-items}. We define a concurrent-move extensive-form game in a
way so that the optimal utility attainable by the leader is equal to the
optimal solution value of the \textsc{Knapsack with unit-items} instance.

The game tree consists of two levels (see Figure~\ref{fig:CM_tree})---the root node consisting of $N$ actions of the leader and $N+1$
actions of the follower.  $M$ denotes a large constant that we use to
force the leader to select a uniform strategy in the root node. More
precisely, we choose $M$ as the smallest integer such that $M > WNv_i$
and $M > Nw_i$ for $i=1,\dots,N$.  In the second level, there is a
state $\calI_i$ corresponding to item $i$ that models the decision of
the leader to include items in the subset (action $\oplus$), or not
(action~$\ominus$).

\begin{figure}[h ]
\centering
root node\\
\begin{tabular}{|c|c|c|c|c|c|}
\hline 
~~~~~ &~~~$f_0$~~~& $f_1$ & $f_2$ &~~~$\ldots$~~~& $f_N$ \\ \hline
$l_1$ & $\calI_1$ & $(0,NM-W-M)$ & $(0,-W-M)$ & $\ldots$ & $(0,-W-M)$ \\ \hline
$l_2$ & $\calI_2$ & $(0,-W-M)$ & $(0,NM-W-M)$ & $\ldots$ & $(0,-W-M)$ \\ \hline
$\vdots$ &  &  &  & $\ddots$ &  \\ \hline
$l_N$ & $\calI_N$ & $(0,-W-M)$ & $(0,-W-M)$ & $\ldots$ & $(0,NM-W-M)$ \\ \hline
\end{tabular}
~\\~\\~\\
$\calI_i$\\
\begin{tabular}{|c|c|c|}
\hline 
~~~~~ &~~~$L$~~~&~~~$R$~~~ \\ \hline
$\oplus$ & $(Nv_i,-Nw_i)$ & $(0,-Nw_i)$ \\ \hline
$\ominus$ & $(Nv_,-Nw_i)$  & $(0,0)$ \\ \hline
\end{tabular}
\caption{Game tree for reduction in proof of Theorem~\ref{th:cm_tree}.}\label{fig:CM_tree}
\end{figure}

Consider a feasible solution $J$ to the \textsc{Knapsack with unit-items} problem with
unit-items. This translates into a strategy for the leader as
follows. In the root node she plays the uniform strategy, and in
sub-game $\calI_i$ plays $\oplus$ with probability 1 if $i \in J$ and
plays $\ominus$ with probability 1 otherwise. We can now observe that
the follower plays $L$ in sub-games $\calI_i$ where $i \in J$, since
ties are broken in favor of the leader, and the follower plays $R$ in
sub-games $\calI_i$ where $i \notin J$. In the root node, action $f_0$
for the follower thus leads to payoff $-\sum_{i \in J} w_i \geq
-W$. Actions $f_k$ for $k \geq 1$ leads to payoff
$$\frac{1}{N}(NM-W-M)+\frac{N-1}{N}(-W-M) = -W.$$ 
Since ties are broken
in favor of the leader, the follower plays action $f_0$, which means
that the leader receives payoff $\sum_{i \in J} v_i$, which is the
value of the \textsc{Knapsack with unit-items} solution.

Consider on the other hand an \emph{optimal} strategy for the
leader. By the structure of the game we have the following lemma.
\begin{claim}
  Without loss of generality the leader plays using a pure strategy in
  each sub-game $\calI_i$.
\end{claim}
\begin{proof}
  If in sub-game $\calI$ the leader commits to playing $\oplus$ with
  probability 1, the follower will choose to play L due to ties being
  broken in favor of the leader. If on the other hand the leader plays
  $\oplus$ with probability strictly lower than 1, the follower will
  choose to play R, leading to utility 0 for the leader, and at most 0
  for the follower. Since the leader can only obtain positive utility
  if the follower plays action $f_0$ in the root node, there is thus no
  benefit for the leader in decreasing the utility for the follower by
  committing to a strictly mixed strategy. In other words, if the leader 
  plays $\oplus$ with probability strictly lower than 1, the leader 
  might as well play $\oplus$ with probability 0. 
\end{proof}

Thus from now on we assume that the leader plays using a pure strategy
in each sub-game $\calI_i$. Let $J \in \calP(N)$ be such the set of
indices $i$ of the sub-games $\calI_i$ where the leader commits to
action $\oplus$.

\begin{claim}
  If the strategy of the leader ensures positive utility it chooses an
  action uniformly at random in the root node.
\end{claim}
\begin{proof}
  Let $\eps_i \in [-\frac{1}{N},1-\frac{1}{N}]$ be such that the
  leader commits to playing action $l_i$ with probability
  $\frac{1}{N}+\eps_i$. Then if the follower plays action $f_0$, the
  leader obtains payoff 
$$\sum_{i \in J}v_i + N\sum_{i \in J}\eps_iv_i$$
  and the follower obtains payoff 
$$-\sum_{i \in J}w_i - N\sum_{i \in J}\eps_iw_i.$$ 
If the follower plays action $f_k$, for $k\geq 1$,
  the leader obtains payoff $0$ and the follower obtains payoff
  $\eps_kNM-W$.
  
  Let $k$ be such that $\eps_k = \max_i{\eps_i}$, and assume to the
  contrary that $\eps_k>0$. Note that
\[
\eps_k \geq \frac{1}{N}\sum_{i: \eps_i>0} \eps_i = -\frac{1}{N}\sum_{i: \eps_i <0} \eps_i \enspace .
\]
We now proceed by case analysis.
\paragraph{Case 1 $\left(\sum_{i\in J} w_i \geq W\right)$:}
By definition of $\eps_k$ and $M$ we have
\[
\begin{split}
\eps_k M & \geq \left(-\frac{1}{N}\sum_{i \in J : \eps_i<0} \eps_i\right)M  > -\frac{1}{N}\sum_{i \in J : \eps_i<0} \eps_i (N w_i) \\& = - \sum_{i \in J : \eps_i<0} \eps_i w_i \geq - \sum_{i \in J} \eps_i w_i
\end{split}
\]
Multiplying both sides of the inequality by $N$ and using the inequality: $-W \geq
-\sum_{i\in J} w_i$, we have
\[
\eps_k NM -W > -\sum_{i\in J} w_i - N\sum_{i \in J} \eps_i w_i \enspace ,
\]
which means that action $f_k$ is preferred by the follower. Thus the
leader receives payoff 0.
\paragraph{Case 2 $\left(\sum_{i\in J} w_i < W\right)$:} Since we have a
\textsc{Knapsack} with unit-items instance, there is a knapsack
solution that obtains value $1+\sum_{i \in J} v_i$, which corresponds
to a strategy for the leader that obtains the same utility. Since the
current strategy is optimal for the leader we must have $\sum_{i \in
  J}v_i + N\sum_{i \in J}\eps_iv_i \geq 1 + \sum_{i \in J} v_i$, which
means that $1\leq N\sum_{i \in J}\eps_iv_i \leq (N^2 \max_i
v_i)\eps_k$, and thus $\eps_k \geq 1/(N^2 \max_i v_i)$. We then have by definition of $M$ that
\[
\eps_kNM-W \geq \frac{NM}{N^2 \max_i v_i}-W > 0 \enspace .
\]
Thus the payoff for the follower is strictly positive for the action
$f_k$, and this is thus preferred to $f_0$, thus leading to payoff 0 to the leader.
\end{proof}

Since there is a strategy for the leader that obtains strictly positive
payoff, we can thus assume that the strategy for the leader chooses an
action uniformly at random in the root node, and the follower chooses
action $f_0$. Since $f_0$ is preferred by the follower to any other
action this means that $\sum_{i \in J} w_i \leq W$, and the leader
obtains payoff $\sum_{i \in J} v_i$. Thus this corresponds exactly to
a feasible solution to the \textsc{Knapsack with unit-items} instance of the same value.
\end{proof}

\section{Appendix: Approximating Optimal Strategies}

In this section we provide the missing details for the algorithms that approximate the optimal strategies for the leader to commit to, for both the behavioral and pure case. \\

\textsc{Theorem \ref{app-behavior}}
\emph{(restated) There is an algorithm that takes as input a turn-based game on a tree
	with chance nodes and a parameter $\epsilon$, and computes a behavioral strategy for the leader. That strategy, combined with some best response of the follower, achieves a payoff that differs by at most $\epsilon$ from the payoff of the leader in a Stackelberg equilibrium in behavioral strategies. The algorithm runs in time $O(\epsilon^{-3}(UH_T)^3T)$, where
	$U=\max_{\sigma,\sigma'}u_1(\sigma)-u_1(\sigma')$, $T$ is the size of the
	game tree and $H_T$ is its height.} \\

We have provided the algorithm in the main body of the paper; its correctness and runtime will follow from the next lemma.

\begin{lemma}
The algorithm of Theorem \ref{app-behavior} is correct and has the desired runtime.
\end{lemma}
\begin{proof}
Recall we are given a turn-based game on a tree with chance nodes and parameter $\epsilon$, and the goal is to compute a behavioral strategy for the leader. We constructed the algorithm in Theorem~\ref{app-behavior} so that it uses dynamic programming to store a table of values for each node in the tree, i.e. a discretized representation of the possible tradeoffs between the utility that the leader can get and the utility that can simultaneously be offered to the follower. The crucial part for proving correctness is arguing that the cumulative error in the leader's utility is bouned additively by the height of the tree. 

For clarity, we repeat the induction hypothesis here. For each sub-tree $T$, the table associated with it, $A_T$, has the following guarantee at each index $k$ in the table:
\begin{enumerate}[a)]
	\item the leader has a strategy for the game tree $T$ that offers the
		follower utility $A_T[k]$ while securing utility {\em at least} $k$ to
		the leader.
	\item no strategy of the leader can (starting from sub-tree $T$) offer the
		follower utility {\em strictly more} than $A_T[k]$, while securing 
		utility {\em at least} $k+H_T$ to the leader, where $H_T$ is the 
		height of the tree $T$.
\end{enumerate}

We now argue this holds for each type of node in the tree. Note that the base case holds trivially by construction, since it is associated with the leaves of the tree. \\

\textbf{Leader nodes} \\

Let $T$ be a leader node, with successors $L$ and $R$, each with tables $A_L$
and $A_R$. If the leader plays $L$ with probability $p$ and plays $R$ with the
remaining probability $(1-p)$, followed up by the strategies that gave the
guarantees for $A_L[i]$ and $A_R[j]$, then the leader would get an expected
$pi+(1-p)j$, while being able to offer $p A_L[i]+(1-p) A_R[j]$ to the follower.
For a given $k$, the optimal combination of the computed tradeoffs becomes:
\begin{align*}
	A_T[k] &:= \max_{i,j,p} \{ p A_L[i]+(1-p) A_R[j]\quad|\quad p i+(1-p) j\geq k\}
\end{align*}

For part 1 of the induction hypothesis, the strategy that guarantees $A_T[k]$ 
simply combines the strategies for the maximizing $A_L[i]$ and $A_R[j]$ along
with the probability $p$ at node $T$. For a given $i$, $j$, and $k$, finding the
optimal value $p$ amounts to maximizing a linear function over an interval,
i.e., it will attain its maximum at one of the end points of the interval. The
table $A_T$ can thus be filled in time $O(n^3)$ by looping over all $0\leq
i,j,k<n$, where $n$ is the number of entries in each table.

For part 2 of the induction hypothesis, assume for contradiction that some strategy $\sigma$ yields utilities 
	$(u_1^\sigma, u_2^\sigma)$ with
	\begin{equation}
		u_1^\sigma\geq k+H_T\quad\text{and}\quad u_2^\sigma>A_T[k]
	\end{equation}
	Let $p_\sigma$ be the probability that $\sigma$ assigns to the action $L$, 
	and let $(u_1^{\sigma,L},u_2^{\sigma,L})$ and $(u_1^{\sigma,R},u_2^{\sigma,R})$
	be the utilities from playing $\sigma$ and the corresponding follower strategy
	in the left and right child respectively. By definition,
	\begin{equation}
		u^\sigma_l=p_\sigma\cdot u_l^{\sigma,L}+(1-p_\sigma)\cdot u_l^{\sigma,R},\quad\forall l\in\{1,2\}
	\end{equation}
	By the induction hypothesis,
	\begin{equation}
		u_2^{\sigma,c}\leq A_c[\floor{u_1^{\sigma,c}}-H_T+1],\quad\forall c\in\{L,R\}
	\end{equation}
	Thus,
	\begin{equation}\label{fptas:leaderineq}
		A_T[k]<u_2^\sigma\leq p_\sigma\cdot A_L[\floor{u_1^{\sigma,L}}-H_T+1]+(1-p_\sigma)\cdot
		A_R[\floor{u_1^{\sigma,R}}-H_T+1]
	\end{equation}
	But
	\begin{align}
		&p_\sigma\cdot (\floor{u_1^{\sigma,L}}-H_T+1)+(1-p_\sigma)\cdot (\floor{u_1^{\sigma,R}}-H_T+1)\\
		&\geq p_\sigma\cdot (u_1^{\sigma,L}-H_T)+(1-p_\sigma)\cdot (u_1^{\sigma,R}-H_T)\\
		&=u_1^\sigma-H_T\geq k
	\end{align}
	meaning that $i=\floor{u_1^{\sigma,L}}-H_T+1$ and 
	$j=\floor{u_1^{\sigma,R}}-H_T+1$ satisfy the constraints in the definition 
	of $A_T[k]$, which contradicts the assumption that $u_2^\sigma>A_T[k]$. \\

\textbf{Chance nodes} \\

Let $T$ be a chance node, with successors $L$ and $R$, each with tables $A_L$
and $A_R$, and let $p$ be the probability that chance picks $L$. If the leader
combines the strategies that gave the guarantees for $A_L[i]$ and $A_R[j]$, then
the leader would get an expected $pi+(1-p)j$ while being able to offer $p
A_L[i]+(1-p) A_R[j]$ to the follower.
For a given $k$, the optimal combination of the computed tradeoffs becomes:
\begin{align*}
	A_T[k] &:= \max_{i,j} \{ p A_L[i]+(1-p) A_R[j]\quad|\quad p i+(1-p) j\geq k\}
\end{align*}

For part 1 of the induction hypothesis, the strategy that guarantees $A_T[k]$ 
simply combines the strategies for the maximizing $A_L[i]$ and $A_R[j]$. The
table $A_T$ can thus be filled in time $O(n^3)$ by looping over all $0\leq
i,j,k<n$, and this can even be improved to $O(n^2)$ by a simple optimization.

For part 2 of the induction hypothesis, assume for contradiction that some strategy $\sigma$ yields utilities 
	$(u_1^\sigma, u_2^\sigma)$ with
	\begin{equation}
		u_1^\sigma\geq k+H_T\quad\text{and}\quad u_2^\sigma>A_T[k]
	\end{equation}
	Let $(u_1^{\sigma,L},u_2^{\sigma,L})$ and $(u_1^{\sigma,R},u_2^{\sigma,R})$
	be the utilities from playing $\sigma$ and the corresponding follower strategy
	in the left and right child respectively. By definition,
	\begin{equation}
		u^\sigma_l=p\cdot u_l^{\sigma,L}+(1-p)\cdot u_l^{\sigma,R},\quad\forall l\in\{1,2\}
	\end{equation}
	By the induction hypothesis,
	\begin{equation}
		u_2^{\sigma,c}\leq A_c[\floor{u_1^{\sigma,c}}-H_T+1],\quad\forall c\in\{L,R\}
	\end{equation}
	Thus,
	\begin{equation}\label{fptas:chanceineq}
		A_T[k]<u_2^\sigma\leq p\cdot A_L[\floor{u_1^{\sigma,L}}-H_T+1]+(1-p)\cdot
		A_R[\floor{u_1^{\sigma,R}}-H_T+1]
	\end{equation}
	But
	\begin{align}
		&p\cdot (\floor{u_1^{\sigma,L}}-H_T+1)+(1-p)\cdot (\floor{u_1^{\sigma,R}}-H_T+1)\\
		&\geq p\cdot (u_1^{\sigma,L}-H_T)+(1-p)\cdot (u_1^{\sigma,R}-H_T)\\
		&=u_1^\sigma-H_T\geq k
	\end{align}
	meaning that $i=\floor{u_1^{\sigma,L}}-H_T+1$ and 
	$j=\floor{u_1^{\sigma,R}}-H_T+1$ satisfy the constraints in the definition 
	of $A_T[k]$, which contradicts the assumption that $u_2^\sigma>A_T[k]$. \\

\textbf{Follower Nodes} \\

Let $T$ be a follower node, with successors $L$ and $R$, each with tables $A_L$
and $A_R$, and let $\tau_L$ and $\tau_R$ be the min-max value for the follower
in $L$ and $R$ respectively. If the leader combines the strategy for $A_L[i]$ in
$L$ with the minmax strategy for $R$, then the followers best response is $L$,
iff $A_L[i]\geq\tau_R$, and similarly it is $R$ if $A_R[j]\geq\tau_L$. Thus, if we let
\begin{align*}
x\downarrow_{\tau}&:=\left\{\begin{matrix}x&\text{, if 
}x\geq\tau\\-\infty&\text{, otherwise}\end{matrix}\right.
\end{align*}
then the optimal combination becomes
\begin{align*}
	A_T[k] &:= \max(A_L[k]\downarrow_{\tau_R}, A_R[k]\downarrow_{\tau_L})\\
\end{align*}

For part 1 of the induction hypothesis, the strategy that guarantees $A_T[k]$ 
simply combines the strategies for the maximizing $A_L[i]$ or $A_R[j]$ in one
branch, and playing minmax in the other. The table $A_T$ can thus be filled in
time $O(n)$.

For part 2 of the induction hypothesis, let $H_T$ be the height of the tree $T$. Suppose that some strategy $\sigma$ yields
	$(u_1^\sigma, u_2^\sigma)$ with
	\begin{equation}
		u_1^\sigma\geq k+H_T\quad\text{and}\quad u_2^\sigma>A_T[k]
	\end{equation}
	Assume wlog.~that the follower plays $L$. Let $(u_1^{\sigma,L},u_2^{\sigma,L})$
	be the utilities from playing $\sigma$ and the corresponding follower strategy
	in the left child. Combined with the induction hypothesis, we get
	\begin{equation}
		A_T[k]<u_2^{\sigma,T}=u_2^{\sigma,L}\leq A_L[\floor{u_1^{\sigma,L}}-H_T+1]=A_T[\floor{u_1^\sigma}-H_T+1]
	\end{equation}
	But this is a contradiction, since $A_T[k]$ is monotonically decreasing in 
	$k$ and 
	\begin{equation}
		\floor{u_1^\sigma}-H_T+1\geq u_1^\sigma-H_T\geq k 
	\end{equation}

From the arguments above, the induction hypothesis holds for all types of nodes, 
and can be computed
in polynomial time in the size of the tree and the number of entries in the
tables.

To complete the proof of the theorem,
let $D=\frac\epsilon{H_T}$ be the initial scaling of the leaders utility.
	Each table for the nodes will now contain
	$n=\frac{U}{D}=\frac{UH_T}{\epsilon}$ entries. Given the tables for the
	successors, each table can be computed in time
	$O(n^3)=O(\epsilon^{-3}(UH_T)^3)$. Since there are $T$ tables in total, we
	get the desired running time.

	Let $A_T$ be the table for the whole tree, and let $i'=\max\{i\ |\
	A_T[i]>-\infty\}$. The strategy associated with $A_T[i']$ guarantees utility
	$i'$ to the leader, and the induction hypothesis directly gives us that no
	strategy of the leader can guarantee more than $i'+H_T$. By dividing by the
	scaling factor $D$, we get that the strategy associated with $A_T[i']$
	guarantees a value that is at most $\epsilon$ lower than that of any other
	strategy.
\end{proof}

If we are only interested in commitment to pure strategies, a very similar
scheme can be constructed---this was stated in Theorem \ref{app-pure}. \\

\textsc{theorem} \ref{app-pure}
\emph{(restated) There is an algorithm that takes as input a turn-based game on a tree
	with chance nodes and a parameter $\epsilon$, and computes a pure strategy
	for the leader. That strategy, combined with some best response of the
	follower, achieves a payoff that differs by at most $\epsilon$ from the
	payoff of the leader in a Stackelberg equilibrium in pure strategies. The
	algorithm runs in time $O(\epsilon^{-2}(UH_T)^2T)$, where
	$U=\max_{\sigma,\sigma'}u_1(\sigma)-u_1(\sigma')$, $T$ is the size of the
	game tree and $H_T$ is its height.} \\

In essence, the algorithm for Theorem \ref{app-pure} is the same, except leader
nodes only consider $p\in\{0,1\}$. The induction hypothesis is the same, except
the quantifications are over pure strategies instead. 
We argue the correctness of this algorithm formally in the following lemma. 

\begin{lemma}
The algorithm for Theorem \ref{app-pure} is correct and has the desired runtime.
\end{lemma}
\begin{proof}
We have the same construction and induction hypothesis as in Theorem \ref{app-behavior}.
Let $T$ be a leader node, with successors $L$ and $R$, each with tables $A_L$
and $A_R$. If the leader plays $L$ (or $R$), followed up by the strategies that gave the
guarantees for $A_L[k]$ (or $A_R[k]$), then the expected leader utility would be
$k$, while being able to offer $A_L[k]$ (or $A_R[k]$ resp.) to the follower.
For a given $k$, the optimal combination of the computed tradeoffs becomes:
\begin{align*}
	A_T[k] &:= \max\{A_c[i]\ |\ i\geq k\land c\in\{L,R\}\}
\end{align*}

For part 1 of the induction hypothesis we simply use the move that maximizes the expression
combined with the strategies that guarantee $A_L[k]$ and $A_R[k]$ in the
successors. The table $A_T$ can thus be filled in time $O(n)$.

For part 2 of the induction hypothesis, assume for contradiction that some pure strategy $\pi$ yields utilities 
	$(u_1^\pi, u_2^\pi)$ with
	\begin{equation}
		u_1^\pi\geq k+H_T\quad\text{and}\quad u_2^\pi>A_T[k]
	\end{equation}
	Assume wlog.~that $\pi$ plays $L$ at $T$, and let $(u_1^{\pi,L},u_2^{\pi,L})$
	be the utilities from playing $\pi$ and the corresponding follower strategy
	in $L$. By definition,
	\begin{equation}
		u^\pi_l=u_l^{\pi,L},\quad\forall l\in\{1,2\}
	\end{equation}
	By the induction hypothesis,
	\begin{equation}
		u_2^{\pi,L}\leq A_L[\floor{u_1^{\pi,L}}-H_T+1]
	\end{equation}
	Thus,
	\begin{equation}\label{fptas:leaderineq}
		A_T[k]<u_2^\pi\leq A_L[\floor{u_1^{\pi,L}}-H_T+1]
	\end{equation}
	But
	\begin{align}
		&\floor{u_1^{\pi,L}}-H_T+1\geq u_1^{\pi,L}-H_T=u_1^\pi-H_T\geq k
	\end{align}
	meaning that $i=\floor{u_1^{\pi,L}}-H_T+1$ satisfies the constraints in the definition 
	of $A_T[k]$, which contradicts the assumption that $u_2^\pi>A_T[k]$.

The proof for the other nodes is identical to those for mixed strategies. The
modified induction hypothesis thus holds for all types of nodes, and can be computed
in polynomial time in the size of the tree and the number of entries in the
tables. The proof of Theorem \ref{app-pure} is very similar to that of Theorem
\ref{app-behavior}, except the calculation has become slightly more efficient.
This completes the proof.
\end{proof}



\end{document}